\documentclass[journal]{IEEEtran}
\usepackage{graphicx}
\usepackage{xspace}
\usepackage{balance}
\usepackage{caption}
\usepackage[labelformat=simple]{subcaption}

\usepackage{subcaption}
\usepackage[export]{adjustbox}

\usepackage{cite}
\ifCLASSINFOpdf
\else
\fi
\usepackage{amsmath}
\newif\ifcomm
\commtrue
\ifcomm
    \newcounter{commentNumberI}
     \newcommand{\Isaac}[1]{\addtocounter{commentNumberI}{1}{{({\color{blue} {(\arabic{commentNumberI}.)} Isaac: #1})}}} %just color
      \newcommand{\Rami}[1]{\addtocounter{commentNumberI}{1}{{({\color{orange} {(\arabic{commentNumberI}.)} Rami: #1})}}} %just color
       \newcommand{\Gal}[1]{\addtocounter{commentNumberI}{1}{{({\color{red} {(\arabic{commentNumberI}.)} Gal: #1})}}} %just color
      \newcommand{\Yoni}[1]{\addtocounter{commentNumberI}{1}{{({\color{blue} {(\arabic{commentNumberI}.)} Yoni: #1})}}} %just color

\else
    \newcommand{\Isaac}[1]{}
    \newcommand{\Gal}[1]{}
    \newcommand{\Yoni}[1]{}
    \newcommand{\Rami}[1]{}

\fi

\usepackage{comment}
\usepackage{graphicx}
\usepackage{bbm,amsmath, graphicx, latexsym, amssymb, amsthm, amsfonts}
\usepackage{xcolor}

\newtheorem{theorem}{Theorem}[section]

\newtheorem{proposition}[theorem]{Proposition}

\newtheorem{assumption}[theorem]{Assumption}

\newtheorem{remark}[theorem]{Remark}
\newtheorem{key observation}[theorem]{Key observation}

\newcommand{\twopartdef}[4]
{
	\left\{
		\begin{array}{ll}
			#1 & \mbox{if } #2 \\
			#3 & \mbox{if } #4
		\end{array}
	\right.
}%%%%%%%%%%%%%%%%%%%%%%%%%%%%%%%%%%%%%%%%%%%%%%

\begin{document}
\title{A Lower Bound on the stability region of Redundancy-$d$ with FIFO service discipline}
\author{Gal~Mendelson
\thanks{Gal~Mendelson is with the Electrical Engineering Faculty, Technion, Israel.}
}
\maketitle

\begin{abstract}
Redundancy-$d$ (R($d$)) is a load balancing method used to route incoming jobs to $K$ servers, each with its own queue. Every arriving job is replicated into $2 \leq d \leq K$ tasks, which are then routed to $d$ servers chosen uniformly at random. When the first task finishes service, the remaining $d{-}1$ tasks are cancelled and the job departs the system. 

Despite the fact that R($d$) is known, under certain conditions, to substantially improve job completion times compared to not using redundancy at all, little is known on a more fundamental performance criterion: what is the set of arrival rates under which the R($d$) queueing system with FIFO service discipline is stable? In this context, due to the complex dynamics of systems with redundancy and cancellations, existing results are scarce and are limited to very special cases with respect to the joint service time distribution of tasks. 

In this paper we provide a non-trivial, closed form lower bound on the stability region of R($d$) for a general joint service time distribution of tasks with finite first and second moments. We consider a discrete time system with Bernoulli arrivals and assume that jobs are processed by their order of arrival. We use the workload processes and a quadratic Lyapunov function to characterize the set of arrival rates for which the system is stable. While simulation results indicate our bound is not tight, it provides an easy-to-check performance guarantee.
\end{abstract}

\begin{IEEEkeywords}
Redundancy Routing, Job Replication, Job Cancellation, Stability, Lyapunov Stability.
\end{IEEEkeywords}

% \fi
\IEEEpeerreviewmaketitle

\section{Introduction}

\IEEEPARstart{R}{edundancy} 
and cancellation based routing has attracted much attention in the last decade \cite{joshi2012coding, dean2013tail, ananthanarayanan2013effective, vulimiri2013low, gardner2017redundancy,gardner2016queueing,raaijmakers2019redundancy}. The basic motivation behind using redundancy and cancellation is reducing the tail of the job completion time distribution. The idea is to replicate a job and send its copies, referred to as \textit{tasks}, to different servers for processing. When the first task is finished being processed the job is deemed complete and leaves the network. The premise is that allowing copies of a job to traverse different paths in the network makes it highly improbable that all of the copies experience large queuing delay and/or processing time. 

Implementing such redundancy and cancellations mechanisms incurs an overhead which can include software, hardware, control, memory and computational power. The performance-cost trade-off of these schemes may well be worthwhile since the potential benefits in terms of performance is known in some cases to be substantial \cite{joshi2012coding, dean2013tail, ananthanarayanan2013effective, vulimiri2013low, gardner2017redundancy}.

In this paper we are concerned with a specific scheme called Redundancy-$d$ (R($d$)), used to route incoming jobs to $K$ servers, working at rate $\mu$, each with its own queue. Within each server, service is given by order of arrival (FIFO). Every arriving job is replicated into $2 \leq d \leq K$ tasks, which are then routed to $d$ distinct servers chosen uniformly at random. When the first task finishes service, the remaining $d{-}1$ tasks are cancelled and the job departs the system. 

Our main research question is concerned with a first order performance criterion of R($d$): what is the set of arrival rates under which the R($d$) system is stable? We refer to this set as \emph{the stability region}. Here and throughout we refer to a system as \emph{stable} if the underlying Markov chain describing the load (e.g. queue lengths, workloads) in the system is positive recurrent.  

For policies with no redundancy such as random routing or `join the shortest queue', it is well known that the queueing system is stable as long as the arrival rate $\lambda$ satisfies $\lambda \in [0,K\mu)$.
For R($d$), denoting by $B_1,\ldots,B_d$ the service time requirements of the $d$ tasks belonging to a single job, such that $\mathbbm{E}[B_i]=\mu^{-1}$, the stability region is known exactly only in two special cases:\\
\noindent (i) $B_1,\ldots,B_d$ are independent and exponentially distributed. Then, the stability region remains $\lambda \in [0,K\mu)$, for all values of $d$ \cite{gardner2016queueing}. \\
\noindent (ii) $d{=}K$ (full redundancy).
 Then the stability region of R($K$) is $\lambda \in [0,1/\mathbbm{E}[\wedge_{i=1}^KB_i])$ \cite{raaijmakers2019redundancy}.

The only other closed form result the author is aware of is a lower bound on the stability region of R($d$) such that if $\lambda \in [0,1/\mathbbm{E}[\wedge_{i=1}^dB_i])$ then the system is stable \cite{raaijmakers2019redundancy}. This lower bound is tight for $d=K$ but is limited in general because it does not depend on $K$.
Finally, the authors of \cite{anton2019stability} implicitly characterize the exact stability region in the case where $B_1,\ldots,B_d$ are identical (i.e. $B_1=\ldots=B_d$) and exponentially distributed. The stability condition is given in terms of the mean number of jobs in service in an associated
`saturated' system. Using this result, our lower bound can be  easily derived for this special case.

Characterizing the stability region of the R($d$) queueing system for $2{\leq} d{<}K$ is a challenging problem. The main difficulty lies in the fact that the amount of work that enters the system upon a job's arrival depends not only on its primitive service time requirement but also on the state of the system. This is due to the possibility of several servers working in parallel on tasks belonging to the same job, not even necessarily beginning at the same time. This results in a Markov chain with state-dependent transition probabilities, a difficult process to analyze. 

In this paper we consider a discrete time system with Bernoulli arrivals and assume that jobs are processed by their order of arrival (FIFO). Our main result is a lower bound on the stability region of R($d$) for a general joint service time distribution of tasks with finite first and second moments. Specifically, we allow a general dependence structure of the service time distribution of tasks belonging to the same job. 

Informally, given $K$, $d$ and the distribution of $(B_1,\ldots,B_d)$, let
\begin{equation*}
    \lambda_{lb}=\frac{K}{\sum_{m=0}^d \Big(\sum_{j=1}^{d-m}\mathbbm{E}[\wedge_{k=1}^jB_k]+m\mathbbm{E}[\wedge_{k=1}^dB_k] \Big)P_m},
\end{equation*}
where 
\begin{equation*}
    \mathbbm{P}_m=\frac{{{K-d}\choose{d-m}  }{{d}\choose{m}}}{{{K}\choose{d}}}.
\end{equation*}
Then if $\lambda \in [0,\lambda_{lb})$ the system is stable. The significance of this result is that it provides a first order performance guarantee in terms of the arrival rates the R($d$) system is able to support. 

In the derivation of $\lambda_{lb}$ we use the fact that given that the system is initially in a balanced state (e.g. empty), the $d$ largest workloads must always be equal \cite{raaijmakers2019redundancy}. For $0{\leq} m{\leq} d$, the quantity $\mathbbm{P}_m$ equals the probability that exactly $m$ of the $d$ servers with the largest workloads are chosen for a job's tasks. The lower bound $\lambda_{lb}$ is derived by bounding the expected amount of incoming work for different values of $m$.

Our bound coincides with the exact stability region for the two extreme cases of $d{=}1$ (no redundancy) and $d{=}K$ (full redundancy). Its relation to the only other known bound of $1/\mathbbm{E}[\wedge_{i=1}^dB_i]$ depends on $K$, $d$ and the distribution of $(B_1,\ldots,B_d)$. 
While no one bound implies the other in the general case, the dependence of $\lambda_{lb}$ on $K$ makes it valuable.

For example, after setting $d=2$ and using the fact that $P_0+P_1=1-P_2$, we obtain
\begin{align*}
    \lambda_{lb}&=\frac{K}{(1-P_2)\mathbbm{E}[B_1]+(1+P_2)\mathbbm{E}[B_1\wedge B_2]}.
\end{align*}
If we choose $K=5$ and $\mathbbm{E}[B_1]=4\mathbbm{E}[B_1\wedge B_2]$ then $P_2=0.1$ and we obtain
\begin{align*}
    \lambda_{lb}= \frac{5}{\frac{9}{10}\mathbbm{E}[B_1]+\frac{11}{10}\mathbbm{E}[B_1\wedge B_2]}> \frac{1}{\mathbbm{E}[\wedge_{k=1}^2 B_k]}.
\end{align*}
However, if we choose $K=3$ and $\mathbbm{E}[B_1]{=}4\mathbbm{E}[B_1\wedge B_2]$, then $P_2=1/3$ and we obtain
\begin{align*}
    \lambda_{lb}= \frac{3}{\frac{2}{3}\mathbbm{E}[B_1]+\frac{4}{3}\mathbbm{E}[B_1\wedge B_2]}< \frac{1}{\mathbbm{E}[\wedge_{k=1}^2 B_k]}.
\end{align*}
Thus, taking the maximum of $\lambda_{lb}$ and $1/\mathbbm{E}[\wedge_{i=1}^dB_i]$ yields a new and improved lower bound for the stability region of R($d$) with FIFO service discipline.

Apart from the closed form lower bound, we make the following additional contributions. The first is a rigorous derivation of a workload model of the R($d$) system. The second is a method of using the standard Lyapunov based technique for proving stability in a setting where the stability region is not a-priori known. The third is identifying a tighter lower bound on the stability region than $\lambda_{lb}$, which we denote by $\lambda_m$, given as a solution to a certain minimization problem. In fact we prove that if $\lambda \in [0,\lambda_m)$ then the system is stable, and then prove $\lambda_{lb}\leq \lambda_m$. Forth and final, we provide simulation results evaluating $\lambda_{lb}$ and comparing it to $1/\mathbbm{E}[\wedge_{k=1}^dB_k]$.

The rest of the paper is organized as follows. In Section \ref{sec:model} we rigorously derive a workload model for the system. In Section \ref{sec:main results} we present our main results and their proofs. Section \ref{sec:sim} is devoted to simultion results.

We use the following notation. For $K \in \mathbbm{N}$ write $[K]=\{1,\ldots,K\}$. For two random vectors $X$ and $Y$ write $X\overset{d}{=}Y$ for equality in distribution. Write $[x]^+$ for $x \vee 0$.
\section{Redundancy routing model}\label{sec:model}
Consider a time slotted system with a single dispatcher and $K$ homogeneous servers. Each server has an infinite size buffer in which a queue can form and the servers do not idle when there is work in the buffer. Each server completes a single unit of service when it has work to do and service is given by order of arrival, \emph{i.e.} FIFO. We assume a server can work on a job that has just arrived. \\

\noindent\textbf{Arrival} At each time slot $t \in \mathbbm{N}$, a job arrives to the dispatcher with probability $0> \lambda<1$, according to the value of a Bernoulli random variable 
(RV) $\mathbbm{1}_A(t)$, such that $\mathbbm{E}[\mathbbm{1}_A(t)]=\lambda$.\\

\noindent\textbf{Routing} When a job arrives, the dispatcher immediately sends $d$ replicas of the job, where $2\leq d \leq K$, to $d$ distinct servers. We refer to these replicas as \emph{tasks}. Denote by  ${\cal{G}}_d$ the set of all $d$-sized subsets of $[K]$. For each $t$, denote by $G_d(t)$ a set-valued RV taking values in ${\cal{G}}_d$ with equal probability. If a job arrives at time slot $t$ then $G_d(t)$ determines which $d$ servers will receive its tasks. We assume that $G_d(t)$ are independent and identically distributed (i.i.d) across time slots. When the first of the job's tasks finishes service, the remaining $d-1$ tasks are immediately cancelled and this marks the job's departure time from the system. For completeness, in the case where tasks are completed at the exact same time, we refer to the task in the smallest indexed server as completed and to the rest as cancelled.\\

\noindent\textbf{Service time distribution.} If a job arrives at time slot $t$, the service duration requirements for its tasks are determined by the random vector $\bar{B}(t){=}(B_1(t),\ldots,B_K(t))$ whose members take values in $\mathbbm{N}$. The quantity $B_i(t)$ represents the service time requirement of the task that is to be sent to server $i$, provided it is a member of $G_d(t)$. Denote $\bar{B}=\bar{B}(1)$.

The homogeneity of the servers is captured by the assumption that the distribution of $\bar{B}$ is symmetric with respect to its members, such that the joint distribution of any subset of $\bar{B}$ coincides with any other subset of the same size. Formally, 

\begin{assumption}[Homogeneity]
For every $k {\in} [K]$, $\{i_1,\ldots,i_k\}{\subset} [K]$ and $\{j_1,\ldots,j_k\}{\subset} [K]$ such that $i_1{<}\ldots{<}i_k$ and $j_1{<}\ldots{<}j_k$ respectively, we assume that
\begin{equation}\label{as:sy}
    (B_{i_1},\ldots,B_{i_k})\overset{d}{=}(B_{j_1},\ldots,B_{j_k}).
\end{equation}
\end{assumption}
For example, direct consequences of this assumption are $\mathbbm{E}[B_1]=\mathbbm{E}[B_2]$ and $\mathbbm{E}[B_1 \wedge B_2]=\mathbbm{E}[B_2 \wedge B_3]$. We also assume $\mathbbm{E}[B_1]<\infty$ and $\mbox{Var}(B_1)<\infty$.

We further assume that $\bar{B}(t)$ are i.i.d. across time slots. As a consequence, the service time requirements of jobs are i.i.d. However, we make no additional assumption on the distribution of $\bar{B}(t)$ for each $t$. Therefore \emph{the service requirements of tasks belonging to the same job may be dependent}. Finally, the arrival and service processes are assumed to be independent.\\

\noindent\textbf{Time scaling.} To avoid the possibility of $\lambda=1$ being inside the stability region, we scale time appropriately. This translates to an assumption on $\bar{B}$ as follows. The average amount of work that enters the system upon a job's arrival is bounded from below by the average of the minimum of the service requirements of its tasks, \emph{i.e.} $\mathbbm{E}[\wedge_{i=1}^d B_i]$. Note that it does not matter which $d$ members we take due to Assumption \ref{as:sy}. The servers clear at most $K$ units of work, so we require  \begin{equation}\label{eq:time scale condition}
\mathbbm{E}\big[\wedge_{i=1}^d B_i\big]>K,
\end{equation}
such that the largest possible $\lambda$ for which the system is stable is strictly smaller than 1. \\

\noindent\textbf{Workload.}
Denote by $W_i(t)$ the workload in buffer $i$ at time slot $t$, \emph{after} a possible arrival and service. This is defined as the amount of time it will take (in time slots) for the existing tasks in the buffer (including the one in service if there is any) to leave the system. Note that tasks can leave the system either due to service completion or due to cancellation. Also, due to the FIFO service discipline, the workloads depend only on present tasks in the buffers and not on future arrivals. This is not true for other service disciplines such as last-in-first-out or processor-sharing. 

Denote $\bar{W}(t){=}(W_1(t),\ldots,W_K(t))$ and $\bar{W}{=}\{\bar{W}(t)\}_{t{\in}\mathbbm{N}}$. We refer to $\bar{W}$ as the \emph{workload process}. We assume that the system starts empty, i.e.
\begin{equation}\label{eq:empty}
W_i(0)=0, \mbox{   } \forall i \in [K].
\end{equation}

%%%%%%%%%%%%%%%%%%%%%%%%%%%%%%%%%55
%%%%%%%%%%%%%%%%%%%%%%%%%%%%%%%%%55
\noindent\textbf{Markov chain formulation.}
We turn to analyze the dynamics of the workload process. To this end,
denote 
\begin{equation}\label{eq:delta}
  \Delta_{i,j}(t)=W_j(t)-W_i(t).
\end{equation}
and let $\mathbbm{1}_i(t)$ denote the indicator RV which equals $1$ if server $i$ is a member of the $d$ chosen servers at time slot $t$.  Then
\begin{equation}\label{eq:i is chosen}
    \mathbbm{E}[\mathbbm{1}_i(t)]=\mathbbm{P}(i\in G_d(t))=\frac{{{K-1}\choose{d-1}}}{{K \choose d}}=\frac{d}{K}.
\end{equation}
\begin{proposition}
For each $i \in [K]$ we have 
\begin{align}\label{eq:dynamics}
    &W_i(t)=[W_i(t-1)+\mathbbm{1}_A(t)A_i(t)-1]^+, \mbox{ where }\cr
    &A_i(t)=\mathbbm{1}_i(t)\Big(\wedge_{j\in G_d(t)} [B_j(t)+\Delta_{i,j}(t{-}1)]^+\Big).
\end{align}
\end{proposition}

\begin{proof}
The first part of  \eqref{eq:dynamics} is a standard balance equation. The workload at time $t$ equals the workload at time $t-1$ plus arrival minus service, and is kept non-negative. The second part of \eqref{eq:dynamics} is less trivial and captures the complexity of the R($d$) model with FIFO service discipline. 

The quantity $A_i(t)$ represents the total amount of work server $i$ receives at time slot $t$ provided a job arrives. If $i{\notin} G_d(t)$ then $\mathbbm{1}_i(t)=0$ and we have $A_i(t)=0$. Otherwise, the amount of work that server $i$ receives depends on the workload in the other $d-1$ members of $G_d(t)$, as well as the service time requirements of the arriving job's tasks. 

Considering the FIFO service discipline within each server and the definition of the workloads, the task that reaches a server first out of the $d$ tasks is the one that is sent to the server with the least workload. Considering the cancellation mechanism, the task that finishes processing first out of the $d$ tasks is the one that is sent to server $j$ for which $W_j(t{-}1)+B_j(t)$ is minimal. Denote this server by $j^*{=}\mbox{argmin}_{j\in G_d(t)}\{W_j(t{-}1)+B_j(t)\}$, where in the case of several minimizers, the smallest index is returned. 

If  $B_{j*}(t){+}W_{j_*}(t{-}1)\leq W_i(t{-}1)$, or, written differently using \eqref{eq:delta}, $B_{j*}(t)+\Delta_{i,j^*}(t{-}1) \leq 0$, the task that arrives to server $i$ will be cancelled before being processed and $A_i(t)=0$. If $B_{j*}(t)+\Delta_{i,j^*}(t{-}1) > 0$, then by the definition of $j^*$ the task in server $i$ will be cancelled at the exact same time server $j^*$ completes its task. Thus, the workload at server $i$ will be truncated and will equal $W_{j^*}(t{-}1){+}B_{j^*}(t)$ and the amount of work server $i$ receives equals 
$$W_{j^*}(t{-}1){+}B_{j^*}(t){-}W_i(t{-}1) =B_{j^*}(t)+\Delta_{i,j^*}(t{-}1).$$
Note that $\Delta_{i,j^*}(t{-}1)$ may be negative. Overall, we obtain
\begin{equation}\label{eq:dynamics 2}
  A_i(t)=\mathbbm{1}_i(t) [B_{j^*}(t)+\Delta_{i,j^*}(t{-}1)]^+.  
\end{equation} 
To connect \eqref{eq:dynamics 2} with \eqref{eq:dynamics}, we argue that
\begin{equation}\label{eq:equiv of j}
    \wedge_{j\in G_d(t)} [B_j(t)+\Delta_{i,j}(t{-}1)]^+=[B_{j^*}(t)+\Delta_{i,j^*}(t{-}1)]^+.
\end{equation}
Indeed, if  $B_{j*}(t)+\Delta_{i,j^*}(t{-}1) \leq 0$, the right hand side of \eqref{eq:equiv of j} is zero, and since $j^*\in G_d(t)$, the left hand side of \eqref{eq:equiv of j} equals zero as well. If  $B_{j*}(t)+\Delta_{i,j^*}(t{-}1) > 0$, by the definition of $j^*$, for $j \in G_d(t)$, we have 
\begin{align*}
    B_j(t)+\Delta_{i,j}(t{-}1)&=B_j(t)+W_j(t{-}1)-W_i(t{-}1)\cr
    &\geq B_{j^*}(t)+W_{j^*}(t{-}1)-W_i(t{-}1)\cr
    &=B_{j*}(t)+\Delta_{i,j^*}(t{-}1)>0,
\end{align*}
which completes the proof.
\end{proof}

 Define the state space $\cal{S}$ of $\bar{W}$ as all members of $\mathbbm{Z}^K$ that can be reached from an empty state. Equations \eqref{eq:empty}-\eqref{eq:dynamics} uniquely define the process $\bar{W}$ as a Markov chain on $\cal{S}$. Since $\mathbbm{P}(\mathbbm{1}_A(t){=}0){>}0$, all states in $\cal{S}$ communicate and the empty state has a self transition. Therefore, $\bar{W}$ is irreducible and a-periodic.

\begin{remark}
Since we have assumed the system starts empty \eqref{eq:empty}, by Property 1 in \cite{raaijmakers2019redundancy}, the $d$ largest workloads are always equal. Thus the largest $d$ components of any member of ${\cal{S}}$ must be equal.
\end{remark}
\begin{remark}
The R($d$) policy routes tasks to $d$ servers chosen uniformly at random. It does not use workload or service time information, which we only use for modelling and stability analysis.
\end{remark}
%%%%%%%%%%%%%%%%%%%%%%%%%%%%%%%%%%%%%%%%%%%%%%%%%%%%%%
%%%%%%%%%%%%%%%%%%%%%%%%%%%%%%%%%%%%%%%%%%%%%%%%%%%%%%
\section{Main results}\label{sec:main results}
We first state our results, then discuss them in detail. The proofs immediately follow.
\subsection{Statements}
Our first result identifies a non-trivial lower bound on the stability region, given as the solution of a certain minimization problem. This lower bound, which we denote by $\lambda_m$, satisfies that if $\lambda \in [0,\lambda_m)$ then $\bar{W}$ is positive recurrent. Our second result, which is the main result of this paper, is a closed form formula for a lower bound $\lambda_{lb}$ on the stability region,  satisfying $\lambda_{lb} \leq \lambda_m$. To this end, define the space of ordered states in ${\cal{S}}$ by 
\begin{equation}\label{eq:def of domain 0}
    {\cal{S}}_0=\{\bar{s}\in {\cal{S}}: s_1\leq   \ldots \leq s_{K-d+1}= \ldots=s_K\},
\end{equation}
and define the space of vectors capturing the differences between the coordinates of members of ${\cal{S}}_0$ by
\begin{align}\label{eq:def of domain}
    {\cal{D}}_{{\cal{S}}_0}=\{\bar{\delta}\in \mathbbm{Z}_{+}^{K-1}:&\exists \bar{s} \in {\cal{S}}_0 \mbox{ such that }  \delta_{i}=s_{i+1}-s_i,\cr
    &\mbox{ for } 1\leq i \leq K-1\}.
\end{align}
Note that $\delta_{K-d+1}=\ldots=\delta_{K-1}=0$. For $\bar{\delta}\in {\cal{D}}_{{\cal{S}}_0}$, define
\begin{equation}\label{eq:def of deltas}
    \delta_{i,j}=\sum_{k=i}^{j-1}\delta_k.
\end{equation}
For ease of notation, denote
\begin{equation}\label{eq:simple notation}
    (G_d,B_j,\mathbbm{1}_A,\mathbbm{1}_i)=(G_d(1),B_j(1),\mathbbm{1}_A(1),\mathbbm{1}_i(1)).
\end{equation}
Define
\begin{equation}\label{eq:la1}
    \lambda_m=\inf_{\bar{\delta}\in {\cal{D}}_{{\cal{S}}_0}}\Bigg\{\frac{K}{\sum_{i \in [K]}\mathbbm{E}\big[\mathbbm{1}_i\big(\wedge_{j\in G_d} [B_j+\delta_{i,j}]^+\big)\big]}\Bigg\}.
\end{equation}

\begin{proposition}\label{thm:max_problem}
If $\lambda \in [0,\lambda_m)$ then $\bar{W}$ is positive recurrent.   
\end{proposition}

We now state the main result of this paper. To this end, for $0\leq m \leq d$, define
\begin{equation}\label{eq:omega m}
    \Omega_m=\{\mbox{choosing m out of d largest workloads}\}.
\end{equation}
A simple calculation yields
\begin{equation}\label{eq:pm}
    \mathbbm{P}_m:=\mathbbm{P}(\Omega_m)=\frac{{{K-d}\choose{d-m}  }{{d}\choose{m}}}{{{K}\choose{d}}}.
\end{equation}
Define
{\small
\begin{equation}\label{eq:la2}
    \lambda_{lb}=\frac{K}{\sum_{m=0}^d \Big(\sum_{j=1}^{d-m}\mathbbm{E}[\wedge_{k=1}^jB_k]+m\mathbbm{E}[\wedge_{k=1}^dB_k] \Big)P_m}
\end{equation}
\par}
\begin{theorem}\label{cor:la2}
Let $\lambda_m$ and $\lambda_{lb}$ defined as in \eqref{eq:la1} and \eqref{eq:la2} respectively. Then
\begin{equation*}
    0 <\lambda_{lb} \leq \lambda_m < 1.
\end{equation*}
Specifically, if $\lambda \in [0,\lambda_{lb})$ then $\bar{W}$ is positive recurrent.
\end{theorem}

%%%%%%%%%%%%%%%%%%%%%%%%%%%%%%%%%%%%%%%%%%%%%%%%
%%%%%%%%%%%%%%%%%%%%%%%%%%%%%%%%%%%%%%%%%%%%%%%%%

\subsection{Intuition and discussion}
Before proving Proposition \ref{thm:max_problem} and Theorem \ref{cor:la2}, we discuss \eqref{eq:la1} and \eqref{eq:la2} in detail.\\

%%%%%%%%%%%%%%%%%%%%%%%%%%%%%%%%%%%%%%%%%%%%%%%%
%%%%%%%%%%%%%%%%%%%%%%%%%%%%%%%%%%%%%%%%%%%%%%%%%

\noindent\textbf{Intuition for $\lambda_m$.} The basic idea is as follows. By \eqref{eq:dynamics}, the average amount of incoming work at time slot $t$ is given by 
{\small
\begin{align}\label{eq:avg work}
    &\mathbbm{E}[\mathbbm{1}_A(t)\sum_{i\in[K]}A_i(t)]\cr
    &=\lambda \mathbbm{E}\Big[\sum_{i\in[K]}\mathbbm{1}_i(t)\Big(\wedge_{j\in G_d(t)} [B_j(t)+\Delta_{i,j}(t{-}1)]^+\Big)\Big].
\end{align}
\par}
This quantity depends on the state at the end of time slot $t-1$ only through $\{\Delta_{i,j}(t{-}1)\}$, i.e. the difference between the workloads. We \emph{require} that for all relevant values of $\{\Delta_{i,j}(t{-}1)\}$ the right hand side of \eqref{eq:avg work} is less than or equal to $K$ and find the largest $\lambda$ for which this still holds. This is a notion of sub-criticality. The challenge is to prove that this is sufficient for stability for the R($d$) system. 

The infimum in \eqref{eq:la1} is taken over ${\cal{D}}_{{\cal{S}}_0}$ and not $\mathbbm{Z}^{K-1}$ for two reasons. First, by the symmetry of the servers, it is sufficient to consider only ordered states in $\cal{S}$ of the form $s_1\leq s_2\leq \ldots \leq s_K$, i.e. $\Delta_{i,j}(t{-}1)\in\mathbbm{Z}_{+}$, whenever $i\leq j$. Second, all states reached by $\bar{W}$ must have that the $d$ largest workloads are equal, namely $\Delta_{K-d+1,j}(t{-}1)=0$ for $j=K-d+1,\ldots,K$. \\

%%%%%%%%%%%%%%%%%%%%%%%%%%%%%%%%%%%%%%%%%%%%%%%%
%%%%%%%%%%%%%%%%%%%%%%%%%%%%%%%%%%%%%%%%%%%%%%%%%

\noindent\textbf{Intuition for $\lambda_{lb}$.} Suppose a job arrives and $d$ corresponding tasks are sent to $d$ distinct servers. Further suppose that exactly $m$ of the servers with the largest workloads are chosen, where $0 \leq m \leq d$. For simplicity, assume that the chosen servers are $\{1,\ldots,d\}$, such that $W_1\leq \ldots \leq W_d$ and that the $m$ servers $d-m+1,\ldots,d$ have the largest workload in the system, implying  $W_{d-m+1}=W_{d-m+2}=\ldots=W_d$. Then, server $1$ receives at most $B_1$ units of work, server 2 at most $B_1 \wedge B_2$ and up to server $d-m$ receiving at most $B_1 \wedge \ldots \wedge B_{d-m}$. The last $m$ servers receive at most $B_1 \wedge \ldots \wedge B_{d}$. Taking expectation and summing over possible values of $m$ gives an upper bound on the expected amount of work that enters the system upon a job's arrival, which yields the closed form expression of $\lambda_{lb}$ in \eqref{eq:la2}.\\

%%%%%%%%%%%%%%%%%%%%%%%%%%%%%%%%%%%%%%%%%%%%%%%%
%%%%%%%%%%%%%%%%%%%%%%%%%%%%%%%%%%%%%%%%%%%%%%%%%

\noindent\textbf{Why we included $\lambda_{m}$ in this paper.}
One can prove that $\lambda_{lb}$ is a lower bound on the stability region without resorting to $\lambda_m$ at all. This requires a minor modification of the proof of Proposition \ref{thm:max_problem} and some elements from the proof of Theorem \ref{cor:la2}. 

However, we feel that $\lambda_m$ is interesting in its own right. Indeed, given $K$, $d$ and the distribution of $\bar{B}$, if one can solve for $\lambda_m$ numerically then one may obtain a tighter bound than $\lambda_{lb}$. But, this is not trivial. First, the set ${\cal{D}}_{{\cal{S}}_0}$ is infinite. Second, different distributions of $\bar{B}$ may change the set of possible states the system can reach resulting in different sets ${\cal{D}}_{{\cal{S}}_0}$ which, in turn, may be difficult to characterize. Third, for each member of ${\cal{D}}_{{\cal{S}}_0}$, one must explicitly calculate the expected value of the minimum of several functions of $\bar{B}$. This may be computationally expensive for certain distributions. We leave this as an open problem.
\\
%%%%%%%%%%%%%%%%%%%%%%%%%%%%%%%%%%%%%%%%%%%%%%%%
%%%%%%%%%%%%%%%%%%%%%%%%%%%%%%%%%%%%%%%%%%%%%%%%%

\noindent\textbf{The special case of $d=1$ (no replication).} In this case $\vert G_d(t) \vert =1$ and every incoming job is randomly routed to one of the servers with equal probability. Together with the fact that by definition $\Delta_{i,i}(t{-}1)=0, \forall i \in [K]$,  Equation \eqref{eq:avg work} reduces to 
\begin{align*}%\label{eq:avg work d1}
    \mathbbm{E}[\mathbbm{1}_A(t)\sum_{i\in[K]}A_i(t)]
    =\lambda \mathbbm{E}\Big[\sum_{i\in[K]}\mathbbm{1}_i(t)B_i(t)\Big]=\lambda \mathbbm{E}\Big[B_1(t)\Big]
\end{align*}
and the solution to \eqref{eq:la1} is given by $\lambda_m=K/\mathbbm{E}[B_1]=K\mu,$
as expected. As for $\lambda_{lb}$, substituting $d=1$ in \eqref{eq:la2} and using the convention that $\sum_{j=1}^0 =0$ yields $\lambda_{lb}=\lambda_m$.\\

%%%%%%%%%%%%%%%%%%%%%%%%%%%%%%%%%%%%%%%%%%%%%%%%
%%%%%%%%%%%%%%%%%%%%%%%%%%%%%%%%%%%%%%%%%%%%%%%%%

\noindent\textbf{The special case of $d=K$ (full replication).} 
In this case the workloads of all of the servers are equal and all servers are chosen for each job. Thus, $\mathbbm{P}_K =1$, and by \eqref{eq:la2}, $\lambda_{lb}=1/\mathbbm{E}[\wedge_{i=1}^KB_i]$ as expected.\\

%%%%%%%%%%%%%%%%%%%%%%%%%%%%%%%%%%%%%%%%%%%%%%%%
%%%%%%%%%%%%%%%%%%%%%%%%%%%%%%%%%%%%%%%%%%%%%%%%%

\noindent\textbf{Example calculation.} 
Consider the case where $K=3$, $d=2$ and the joint service time distribution $\bar{B}=(B_1,B_2,B_3)$ satisfies that $B_1,B_2 \text{ and } B_3$ are i.i.d and 
\begin{equation*}
    B_1=
    \begin{cases}
    \alpha & \text{with probability } p \\
    \beta & \text{with probability } 1-p ,
    \end{cases}
\end{equation*}
where $0\leq p \leq 1$ and $\alpha,\beta \in \mathbbm{N}$. By \eqref{eq:pm} we have 
$P_0=0 , P_1=2/3 \text{ and } P_2=1/3,$
and a straightforward calculation yields
\begin{align*}%\label{eq:example E}
    &\mathbbm{E}[B_1]=\alpha p+\beta (1-p) \cr
    &\mathbbm{E}[B_1\wedge B_2]=\alpha (1-(1-p)^2)+\beta (1-p)^2.
\end{align*}
Thus, by \eqref{eq:la2} we have
\begin{align*}
    \lambda_{lb}&=\frac{3}{(\mathbbm{E}[B_1]+\mathbbm{E}[B_1\wedge B_2])P_1+2\mathbbm{E}[B_1\wedge B_2]P_2}\cr
    &=\frac{3}{\frac{2}{3}(\mathbbm{E}[B_1]+2\mathbbm{E}[B_1\wedge B_2])}\cr
    &=\frac{9}{2\alpha  p (5-2p)+2\beta(1-p)(3-2p)}.
\end{align*}

We use the following two results in the proofs that follow. 

\begin{proposition}[Balance]\label{prop:balance}
In a R($d$) system, with $\bar{0}$ initial condition, the $d$ largest workloads are equal at all times. 
\end{proposition}
\begin{proof}
The proof is given in Property 1 in \cite{raaijmakers2019redundancy}.
\end{proof}

\begin{proposition}[Average workload]\label{prop:average workload}
Consider a R($d$) system with $K$ servers. Let $\bar{s}=(s_1,\ldots,s_K)\in {\cal{S}}$ such that $s_1\leq\ldots\leq s_K$. Let $s_{i,j}=s_j-s_i$. Denote
\begin{equation*}%\label{eq:f_notation}
f_i(\bar{s})=\mathbbm{E}\big[\mathbbm{1}_i\big(\wedge_{j\in G_d} [B_j+s_{i,j}]^+\big)\big]
\end{equation*}
Then
\begin{align}\label{eq:avg_work}
    \mathbbm{E}[A_i(1) \mid \bar{W}(0)=\bar{s}]=f_i(\bar{s}),
\end{align}
and
\begin{equation}\label{eq:monotone_f_i}
f_1(\bar{s})\geq \ldots \geq f_K(\bar{s}).
\end{equation}
\end{proposition}

\begin{proof}
Relation \eqref{eq:avg_work} is an immediate consequence of \eqref{eq:dynamics} and \eqref{eq:simple notation}.  We now prove that for any $i\in \{1,\ldots,K-1\}$ we have $f_i(\bar{s})\geq f_{i+1}(\bar{s})$ and the relation \eqref{eq:monotone_f_i} follows. This simply means that the average incoming workload is monotonic non-increasing when considering servers ordered by their workload.

Consider the four possibilities describing whether or not $i$ and $i+1$ are members of $G_d$. Under the event that $i$ and $i+1$ are not in $G_d$, both receive zero work. If both are in $G_d$, then $i+1$ cannot receive more work than $i$ due the the minimum taken in \eqref{eq:dynamics}. Finally, for any event under which $i$ is chosen and $i+1$ is not, there is an event with equal probability where $i$ is not chosen and $i+1$ is, and the rest $d-1$ servers stay the same, and vice versa. Again, by \eqref{eq:dynamics}, the amount of average work that enters server $i$ under the first event is no less than the amount of average work that enters server $i+1$ under the second event. This concludes the proof. 
\end{proof}

%%%%%%%%%%%%%%%%%%%%%%%%%%%%%%%%%%%%%%%%%%%%555
%%% lambda 2 %%%%%%%%%%%%%%%%%%%%%%%%
\subsection{Proofs of main results}

\noindent \textbf{Proof of Proposition \ref{thm:max_problem}}. Since $\bar{W}$ is irreducible and aperiodic, by Theorem 3.3.7 of \cite{srikant2013communication}, it suffices to prove that if $\lambda{<}\lambda_m$, then a Lyapunov drift condition holds. Namely, that there exist a function ${\cal{L}}: {\cal{S}} \rightarrow  \mathbbm{R}_+$, a finite set $F \subset {\cal{S}}$ and constants $\epsilon$, $C_1>0$ such that
\begin{align}\label{eq:drift}
\mathbbm{E}[ \Delta{\cal{L}}(t{+}1)\mid \bar{W}(t)=\bar{s}]\leq\twopartdef{-\epsilon}{\bar{s} \in {\cal{S}}\setminus F}{C_1}{\bar{s} \in F}    
\end{align} 
where $\Delta{\cal{L}}(t{+}1)$ denotes the drift at time slot $t+1$, namely
\begin{equation}\label{eq:deltaL}
    \Delta{\cal{L}}(t{+}1)={\cal{L}}(\bar{W}({t{+}1}))-{\cal{L}}(\bar{W}(t)).
\end{equation} 
While we assumed in \eqref{eq:empty} that the system starts empty, with a slight abuse of notation and for simplicity, in what follows we suppress the dependence on $t$ by writing $\bar{W}(0)$ and $\bar{W}(1)$ instead of $\bar{W}(t)$ and $\bar{W}(t+1)$, respectively. We choose the quadratic function 
\begin{equation}\label{eq:lyap fn}
    {\cal{L}}(\bar{s}){=}\sum_{i=1}^Ks_i^2.
\end{equation}
Since the members of $\bar{B}$ have finite first and second moments, it is trivial that the left hand side of \eqref{eq:drift} is bounded from above by some positive constant $C_1$, uniformly over all $\bar{s} \in F$, for any finite set $F \subset {\cal{S}}$. We omit the details. 

Next, we choose the finite set $F$ to be of the form
\begin{equation}\label{eq:F}
    F=\{\bar{s}=(s_1,\ldots,s_K) \in {\cal{S}} :\max_{1\leq i \leq K}s_i<C_2 \},
\end{equation}
where $C_2>0$ is a large enough constant whose value is determined later in the proof. 

Consider a state $\bar{s} \in {\cal{S}} \setminus F$. By the homogeneity of the servers, the symmetric distribution of $\bar{B}$ and the uniformly at random routing choice, we can, without loss of generality, consider $\bar{s}=(s_1,\ldots,s_K)$ such that 
\begin{equation*}%\label{eq:order}
    0\leq s_1 \leq \ldots \leq s_K.
\end{equation*}
By the definition of $F$ in \eqref{eq:F}, we have
\begin{equation}\label{eq:max}
    s_K>C_2,
\end{equation}
and by Proposition \ref{prop:balance},
\begin{equation*}%\label{eq:balance}
    s_{K-d+1}=\ldots=s_K.
\end{equation*}
By \eqref{eq:dynamics}, using the notation in \eqref{eq:simple notation} and denoting $A_i=A_i(1)$, we have 
\begin{equation*}
    W_i(1)=[s_i+\mathbbm{1}_A A_i-1]^+
\end{equation*}
and therefore
\begin{equation}\label{eq:get rid of plus}
    W_i^2(1)\leq s_i^2+2s_i(\mathbbm{1}_A A_i-1)+(\mathbbm{1}_A A_i-1)^2.
\end{equation}
Using \eqref{eq:avg_work}, \eqref{eq:deltaL}, \eqref{eq:lyap fn} and \eqref{eq:get rid of plus} we obtain
{\small
\begin{align}\label{eq:drift calc}
\mathbbm{E}[ \Delta{\cal{L}}(1)\mid & \bar{W}(0){=}\bar{s}]=\sum_{i =1}^K\mathbbm{E}\Big[ W_i^2(1)-W_i^2(0)\mid \bar{W}(0)=\bar{s}\Big]\cr
&\leq 
2\sum_{i =1}^K s_i\Big(\lambda\mathbbm{E}[A_i{\mid} \bar{W}(0)=\bar{s}]-1\Big)+C_3 \cr
&=2\sum_{i =1}^K s_i\Big(\lambda \mathbbm{E}[\mathbbm{1}_i]f_i(\bar{s})-1\Big)+C_3,
\end{align} 
\par}
where the constant $C_3>0$ satisfies
\begin{equation}\label{eq:C_3}
    \sum_{i \in [K]}\mathbbm{E}[(\mathbbm{1}_A A_i{-}1)^2]\leq \sum_{i \in [K]}\mathbbm{E}[(\sum_{i=1}^dB_i)^2]\leq C_3.
\end{equation}
The existence of $C_3$ is due to the finite first and second moments of the members of $\bar{B}$. 

Denote 
\begin{equation}\label{eq:def of d}
  d_i=\lambda \mathbbm{E}[\mathbbm{1}_i]f_i(\bar{s})-1.   
\end{equation}
The argument proceeds by analyzing $\sum_{i=1}^K s_i d_i$. 
First, by \eqref{eq:i is chosen} and \eqref{eq:monotone_f_i}, 
\begin{equation*}
    d_1 \geq \ldots \geq d_K.
\end{equation*} Second, since $\lambda<\lambda_m$, there exists $\epsilon_0>0$ such that $\lambda=\lambda_m-\epsilon_0$. Thus
\begin{align*}
    \sum_{i =1}^K d_i&=\lambda \sum_{i =1}^K \mathbbm{E}[\mathbbm{1}_i]f_i(\bar{s})-K \cr
    &=(\lambda_m-\epsilon_0) \sum_{i =1}^K \mathbbm{E}[\mathbbm{1}_i]f_i(\bar{s})-K\cr
    &\leq-\epsilon_0 \sum_{i =1}^K \mathbbm{E}[\mathbbm{1}_i]f_i(\bar{s})<0,
\end{align*}
where in the last inequality we have used the definition of $\lambda_m$ in \eqref{eq:la1}.
Denote by $k_0$ the lowest index in $[K]$ such that  $\sum_{i =1}^{k_0} d_i<0$, namely 
\begin{equation}\label{eq:k_0}
    k_0=\min{\{j \in [K]:\sum_{i =1}^{j} d_i<0\}}.
\end{equation}
Hence
\begin{equation}\label{eq:positive}
    \sum_{i =1}^{j} d_i \geq 0, \quad \forall j < k_0
\end{equation}
and
\begin{equation}\label{eq:d are negative}
    0>d_{k_0}\geq \ldots \geq d_{K}
\end{equation}

Next, we argue by induction that 
\begin{equation}\label{eq:ind}
    \sum_{i=1}^{j-1}s_i d_i\leq s_{j}\sum_{i=1}^{j-1} d_i,\quad \forall j\in \{1,\ldots,k_0\},
\end{equation}
with the convention that $\sum_{i=1}^{-1}=0$. Inequality \eqref{eq:ind} holds trivially for $j=1$. Suppose it holds for $j=j_0<k_0$. Then  
\begin{align*}
    \sum_{i=1}^{(j_0+1)-1}s_i d_i&= \sum_{i=1}^{j_0-1}s_i d_i+s_{j_0}d_{j_0}\leq
    s_{j_0}\sum_{i=1}^{j_0-1} d_i+s_{j_0}d_{j_0} \cr
    &= s_{j_0}\sum_{i=1}^{j_0} d_i\leq  s_{j_0+1}\sum_{i=1}^{(j_0+1)-1} d_i,
\end{align*}
where the first inequality is due to the induction hypothesis and the second is due to \eqref{eq:positive}. Taking $j=k_0$ in \eqref{eq:ind} yields
\begin{equation*}
    \sum_{i=1}^{k_0-1}s_i d_i\leq s_{k_0}\sum_{i=1}^{k_0-1} d_i
\end{equation*}
and therefore
\begin{align}\label{eq:almost done}
    \sum_{i=1}^{K}s_i d_i&\leq s_{k_0}\sum_{i=1}^{k_0-1} d_i+\sum_{i=k_0}^{K}s_i d_i\cr
    &=s_{k_0}\sum_{i=1}^{k_0} d_i+\sum_{i=k_0+1}^{K}s_i d_i.
\end{align}
By the definition of $k_0$ in \eqref{eq:k_0} and by \eqref{eq:d are negative}, the coefficients $\sum_{i=1}^{k_0} d_i,d_{k_0+1}, \ldots,d_K$, multiplying $s_{k_0},s_{k_0+1},\ldots,s_{K}$ respectively, are strictly negative. Define
\begin{equation}\label{eq:gamma}
   \gamma=\min{\Big\{ \sum_{i=1}^{k_0} d_i,d_{k_0+1}, \ldots,d_K\Big\}}<0 .
\end{equation}
By \eqref{eq:almost done} and \eqref{eq:gamma} we obtain
\begin{equation}\label{eq:done}
    \sum_{i=1}^{K}s_i d_i\leq \gamma (s_{k_0}+s_{k_0+1}+\ldots+s_k)\leq \gamma s_k \leq \gamma C_2,
\end{equation}
where in the last inequality we have used \eqref{eq:max}.
Combining \eqref{eq:drift calc}, \eqref{eq:def of d} and \eqref{eq:done} we obtain
\begin{equation*}
    \mathbbm{E}[ \Delta{\cal{L}}(1)\mid  \bar{W}(0){=}\bar{s}]\leq 2\gamma C_2+C_3,
\end{equation*}
where we recall from \eqref{eq:gamma} that $\gamma<0$. 

Finally, we determine $\epsilon$ and $F$ in \eqref{eq:drift}. Given the primitive arrival and service processes, the constants $C_3$ in \eqref{eq:C_3} and $\gamma$ in \eqref{eq:gamma} are given. Fix some $\epsilon>0$ and choose $C_2$ (which defines $F$) to be large enough such that $2\gamma C_2+C_3<-\epsilon$. This concludes the proof.

\qed
%%%%%%%%%%%%%%%%%%%%%%%%%%%%%%%%%%%%%%%%%%%%%%
%%%%%%%%%%%%%%%%%%%%%%%%%%%%%%%%%%%%%%%%%%%%%%

\noindent \textbf{Proof of Theorem \ref{cor:la2}}.
Recall that by \eqref{eq:la1}
\begin{equation*}
    \lambda_m=\inf_{\bar{\delta}\in {{\cal{D}}_{{\cal{S}}_0}}}\Bigg\{\frac{K}{\sum_{i \in [K]}\mathbbm{E}\big[\mathbbm{1}_i\big(\wedge_{j\in G_d} [B_j+\delta_{i,j}]^+\big)\big]}\Bigg\},
\end{equation*}
where ${\cal{D}}_{{\cal{S}}_0}$ is given in \eqref{eq:def of domain}.
Denote by $G_d(k)$ the $k$th member of $G_d$ such that 
\begin{equation}\label{eq:def of gd1}
    G_d(1)<\ldots<G_d(d).
\end{equation}
With this notation at hand we can write
\begin{align}\label{eq:get g in there}
    \sum_{i \in [K]}\mathbbm{E}\big[\mathbbm{1}_i&\big(\wedge_{j\in G_d} [B_j+\delta_{i,j}]^+\big)\big]\cr
    &=\sum_{k=1}^d\mathbbm{E}\big[\wedge_{j=1}^d [B_{G_d(j)}+\delta_{G_d(k),G_d(j)}]^+\big].
\end{align}
%%%%%%%%%%%%%%%%%%%%%%%%%%%%%%%%%%%%%%%%%%%%%%
\noindent \textbf{Proof that $\lambda_m<1$}.
Taking only the first term of the right hand side of \eqref{eq:get g in there} yields
\begin{align}\label{eq:take only first term}
\sum_{k=1}^d\mathbbm{E}\big[&\wedge_{j=1}^d [B_{G_d(j)}+\delta_{G_d(k),G_d(j)}]^+\big]\cr
&\geq \mathbbm{E}\big[\wedge_{j=1}^d [B_{G_d(j)}+\delta_{G_d(1),G_d(j)}]^+\big].
\end{align}
By the definition of ${\cal{D}}_{{\cal{S}}_0}$, $\delta_{i,j}$ and $\{G_d(1),\ldots,G_d(d)\}$ in  \eqref{eq:def of domain}, \eqref{eq:def of deltas} and \eqref{eq:def of gd1}, respectively, we have that $\delta_{G_d(1),G_d(j)}{\geq} 0$. Therefore 
\begin{align}\label{eq: uniform bound}
\mathbbm{E}\big[\wedge_{j=1}^d [B_{G_d(j)}+\delta_{G_d(1),G_d(j)}]^+\big]&\geq \mathbbm{E}\big[\wedge_{j=1}^d B_{G_d(j)}\big]\cr
&=\mathbbm{E}\big[\wedge_{j=1}^d B_{j}\big],
\end{align}
where the last transition is due to the symmetry assumed in Assumption \eqref{as:sy} and the fact that $G_d$ and $\bar{B}$ are independent.
Combining \eqref{eq:get g in there}, \eqref{eq:take only first term} and \eqref{eq: uniform bound} yields 
\begin{align*}%\label{eq:bound work from below}
\sum_{i \in [K]}\mathbbm{E}\big[\mathbbm{1}_i&\big(\wedge_{j\in G_d} [B_j+\delta_{i,j}]^+\big)\big]]
\geq \mathbbm{E}\big[\wedge_{j=1}^d B_{j}\big]>0.
\end{align*}
Since this bound holds for every $\bar{\delta}\in {{\cal{D}}_{{\cal{S}}_0}}$, using \eqref{eq:la1} we obtain
$\lambda_m \leq K/\mathbbm{E}\big[\wedge_{j=1}^d B_{j}\big]<1,$
where the last transition is due to the time scaling assumption in \eqref{eq:time scale condition}. \qed \\

%%%%%%%%%%%%%%%%%%%%%%%%%%%%%%%%%%%%%%%%%%%%%%
\noindent \textbf{Proof that $\lambda_{lb}\leq\lambda_m$}.
Fix $\bar{\delta}\in {{\cal{D}}_{{\cal{S}}_0}}$.
On $\Omega_m$ in \eqref{eq:omega m}, exactly $m$ out of the $d$ largest workloads are members of $G_d$ and are given by $\{G_d(d-m+1),\ldots,G_d(d)\}$. 
By \eqref{eq:def of domain}, \eqref{eq:def of deltas} and \eqref{eq:def of gd1} we have
\begin{equation*}%\label{eq:delta is zero}
    \delta_{G_d(k),G_d(j)} \leq 0, \mbox{ for } k>\min\{ j,d-m\}.
\end{equation*}
So, for $1\leq k \leq d-m$, we have
\begin{align}\label{eq:key move 1}
\big(\wedge_{j=1}^d [B_{G_d(j)}&+\delta_{G_d(k),G_d(j)}]^+\big)\mathbbm{1}_{\Omega_m}\cr
&\leq\big(\wedge_{j=1}^k [B_{G_d(j)}+\delta_{G_d(k),G_d(j)}]^+\big)\mathbbm{1}_{\Omega_m}\cr
&\leq\big(\wedge_{j=1}^k B_{G_d(j)}\big)\mathbbm{1}_{\Omega_m}
\end{align}
and for $d-m+1\leq k \leq d$ we have 
\begin{align}\label{eq:key move 2}
\big(\wedge_{j=1}^d [B_{G_d(j)}&+\delta_{G_d(k),G_d(j)}]^+\big)\mathbbm{1}_{\Omega_m}\cr
&\leq\big(\wedge_{j=1}^d B_{G_d(j)}\big)\mathbbm{1}_{\Omega_m}.
\end{align}
Using \eqref{eq:get g in there}, \eqref{eq:key move 1} and \eqref{eq:key move 2} we obtain
\begin{align}\label{eq:almost done 2}
    &\sum_{i \in [K]}\mathbbm{E}\big[\mathbbm{1}_i\big(\wedge_{j\in G_d} [B_j+\delta_{i,j}]^+\big)\big]\cr
    &=\mathbbm{E}\Big[\sum_{k=1}^d\big(\wedge_{j=1}^d [B_{G_d(j)}+\delta_{G_d(k),G_d(j)}]^+\big)\Big]\cr
    &=\sum_{m=0}^d\mathbbm{E}\Big[\sum_{k=1}^d\big(\wedge_{j=1}^d [B_{G_d(j)}+\delta_{G_d(k),G_d(j)}]^+\big)\mathbbm{1}_{\Omega_m}\Big]\cr
    &\leq \sum_{m=0}^d \Big(\sum_{k=1}^{d-m}\mathbbm{E}[\wedge_{j=1}^kB_j]+m\mathbbm{E}[\wedge_{j=1}^dB_j] \Big)P(\Omega_m),
\end{align}
where in the last inequality we used the fact that $\mathbbm{1}_{\Omega_m}$ is independent of $\bar{B}$ and $G_d$.
Finally, since the bound in \eqref{eq:almost done 2} is finite and uniform over $\bar{\delta}\in {{\cal{D}}_{{\cal{S}}_0}}$, we obtain
\begin{align*}
    \lambda_m&=\frac{K}{\sup_{\bar{\delta}\in {{\cal{D}}_{{\cal{S}}_0}}}\big\{\sum_{i \in [K]}\mathbbm{E}\big[\mathbbm{1}_i\big(\wedge_{j\in G_d} [B_j+\delta_{i,j}]^+\big)\big]\big\}}\cr
    &\geq \frac{K}{ \sum_{m=0}^d \Big(\sum_{k=1}^{d-m}\mathbbm{E}[\wedge_{j=1}^kB_j]+m\mathbbm{E}[\wedge_{j=1}^dB_j] \Big)P(\Omega_m)}\cr
    &=\lambda_{lb},
\end{align*}
which concludes the proof. \qed \\

%%%%%%%%%%%%%%%%%%%%%%%%%%%%%%%%%%%%%%%%%%%%%%%%%%%%%%
%%%%%%%%%%%%%%%%%%%%%%%%%%%%%%%%%%%%%%%%%%%%%%%%%%%%%%
\section{Simulation}\label{sec:sim}
In this section we present simulation results which shed some light on the behaviour of the stability region of R($d$), our lower bound $\lambda_{lb}$ and the known lower bound $1/\mathbbm{E}[\wedge_{k=1}^dB_k]$.

We consider the R($d$) system with $K=10$ servers, working according to the FIFO service discipline. The service time distribution of tasks, $\bar{B}$, is comprised of i.i.d  random variables $B_1,\ldots,B_{10}$ such that 
\begin{equation*}
    B_1=
    \begin{cases}
    10 & \text{with probability } 0.9 \\
    100 & \text{with probability } 0.1. 
    \end{cases}
\end{equation*}
Since $B_1\geq 10$, the time scaling condition \eqref{eq:time scale condition} holds, and thus the stability region is a subset of $[0,1]$ for all values of $d$. For each value of $d\in\{1,\ldots,10\}$ we run simulations on a large number of time slots for different loads (namely, values of the arrival rate $\lambda$) in $[0,1]$. The number of time slots was chosen such that the difference in the outputs of different runs at the maximal load were negligible.

For each simulation run corresponding to a specific $(d,\lambda)$ pair, we calculate the running average workload in the system (over all time-slots, after an initial duration required for convergence). Whenever the Markov chain is positive recurrent (i.e. the system is stable), it is also ergodic. Thus
the running average workload converges, and one simulation run is enough to calculate the steady-state average workload. 

The idea is that for values of $d$ where the stability region is not known, the steady state average workload dependence on $\lambda$, and, specifically, for what loads it becomes very large, suffices as an approximation for the actual stability region. We also calculate our lower bound $\lambda_{lb}$ and the known lower bound $1/\mathbbm{E}[\wedge_{k=1}^dB_k]$. Figure \ref{fig:stability region 10 servers} depicts the results. 

%%%%%%%%%%%%%%%%%%%%%%%%%%%%%%%%%55
%%%%%%%%%%%%%%%%%%%%%%%%%%%%%%%%%%%
\begin{figure*}%{\linewidth} \centering
\includegraphics[width=0.6\textwidth,height=7cm,center]{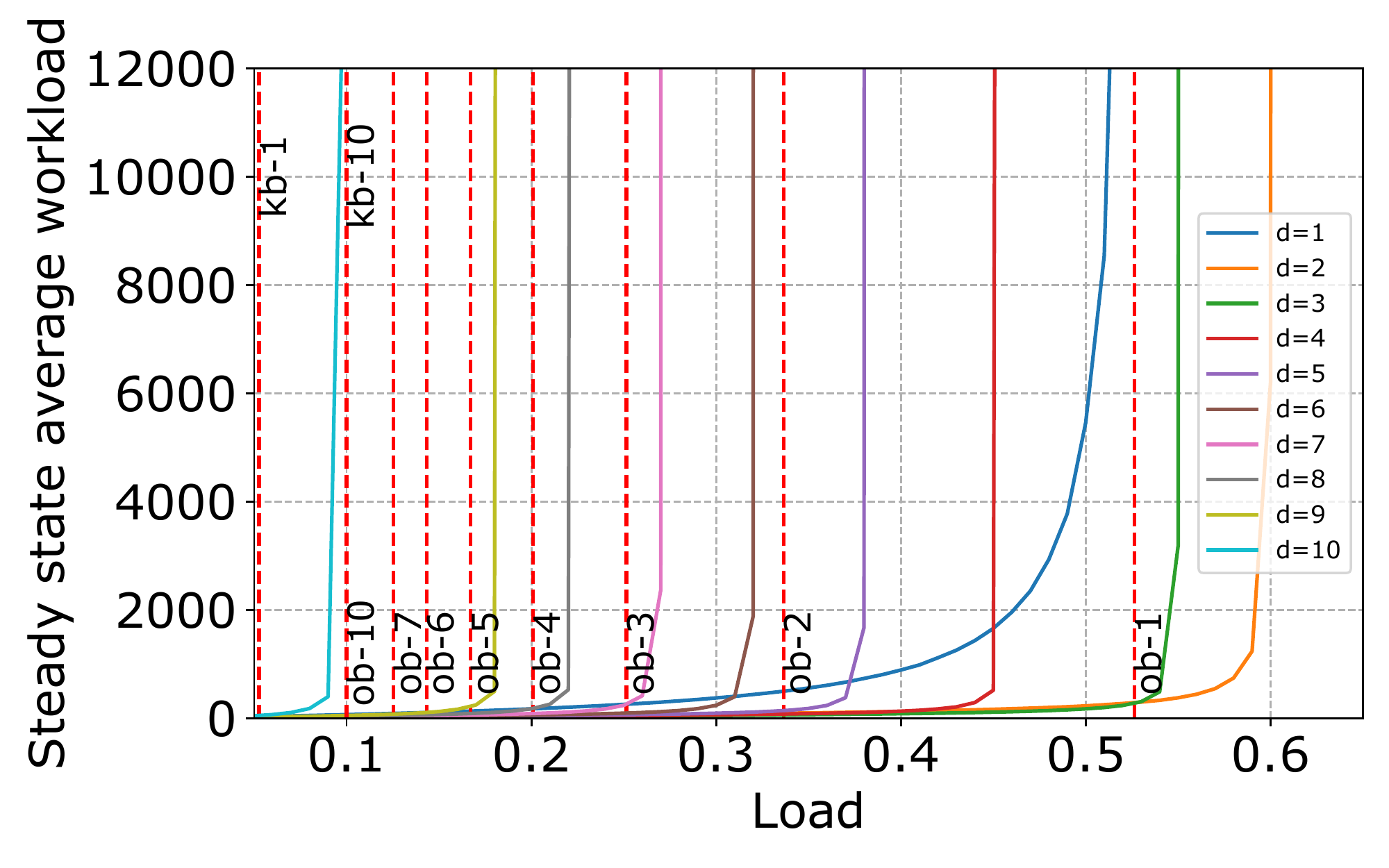}
\caption{\footnotesize Steady state average workload vs. load (values of the arrival rate $\lambda$) for $K=10$ servers and different values of $d$. The vertical lines marked `ob-$i$' and `kb-$i$' stand for 'our lower bound' $\lambda_{lb}$ and the 'known lower bound' of $1/\mathbbm{E}[\wedge_{k=1}^dB_k]$, respectively, for $d=i$. Only kb-1 and kb-10 are shown because all other values of kb are between them. The values of ob-8 and ob-9 are not shown and satisfy ob-$10<$ob-$9<$ob-$8<$ob-7.
}
\label{fig:stability region 10 servers}
\end{figure*}
%%%%%%%%%%%%%%%%%%%%%%%%%%%%%%%%%%55
%%%%%%%%%%%%%%%%%%%%%%%%%%%%%%%%%%%%

The stability region for $d=1$ is marked by the vertical line `ob-$1$' (which stands for `our bound') and equals approximately 0.52. The simulation indicates that the stability region for $d=2$ is the largest and equals approximately 0.6. The stability region for $d=3$ is still larger than that for $d=1$ and equals approximately 0.56. For larger values of $d$ the stability region decreases substantially until reaching around 0.1 for $d=10$. The non-monotone behaviour of the stability region with respect to the values of $d$ is evidence for why it is challenging to study it. 

On the one hand, our bound is not tight. For example, ob-$2$ marks our bound for $d=2$ and equals almost half of the actual stability region. On the other hand, it is much better than the known lower bound for all $1 \leq d<K$. Another interesting result of the simulation is that the steady state average workload decreases substantially for $d>1$ compared to the case where $d=1$ (no replication). In fact, it can be seen that most of the improvement is achieved by using $d=2$ instead of $d=1$. If the system under consideration is currently working at around 0.3 load, then our lower bound guarantees that the system remains stable for $d=2$ while obtaining the benefits of replication.  

Next, to further compare the lower bounds, we consider the R($d$) system with $K=30$ servers. Instead of choosing a specific distribution for $\bar{B}$, we specify the connection between the different expected values needed to calculate the lower bounds. Figure \ref{fig:ev ratio 1} depicts the results for the case where 
\begin{equation*}
    \mathbbm{E}[\wedge_{k=1}^dB_k]=K/d^{0.5}, \quad d \in \{1,\ldots,K\},
\end{equation*}
and Figure \ref{fig:ev ratio 1} depicts the results for the case where 
\begin{equation*}
    \mathbbm{E}[\wedge_{k=1}^dB_k]=2K/d^{1.1}, \quad d \in \{1,\ldots,K\}.
\end{equation*}
As mentioned in the introduction, no one bound implies the other and their values highly depend on the distribution of $\bar{B}$ and the value of $d$. Taking the maximum of the lower bounds yields a new and improved lower bound.
%%%%%%%%%%%%%%%%%%%%%%%%%%%%%%%%%%%%%
%%%%%%%%%%%%%%%%%%%%%%%%%%%%%%%%%%%%
\begin{figure}[!t]
\centering

\begin{subfigure}[c]{\linewidth} \centering
\includegraphics[width=0.99\textwidth]{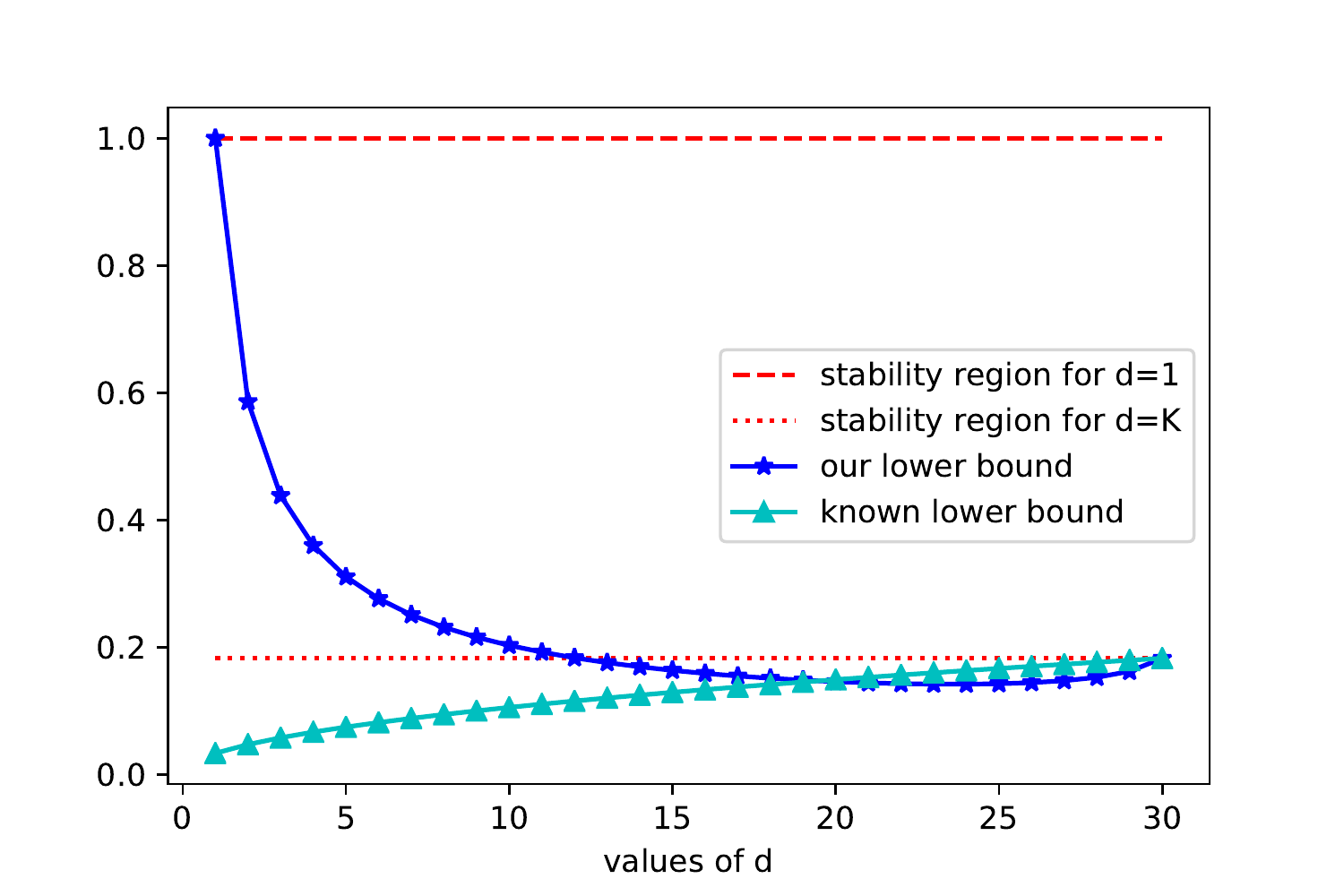}
\caption{\footnotesize Our lower bound and the known lower bound for $K=30$ and different values of $d$ for the case where $\mathbbm{E}[\wedge_{k=1}^dB_k]=K/d^{0.5}$ 
}
\label{fig:ev ratio 1}
\end{subfigure}

\begin{subfigure}[c]{\linewidth} \centering
%\MyIncludeGraphics[clip, trim=0 10.3cm 0 10.5cm, width=0.99\linewidth]{figs/fix_3.pdf}
\includegraphics[width=0.99\linewidth]{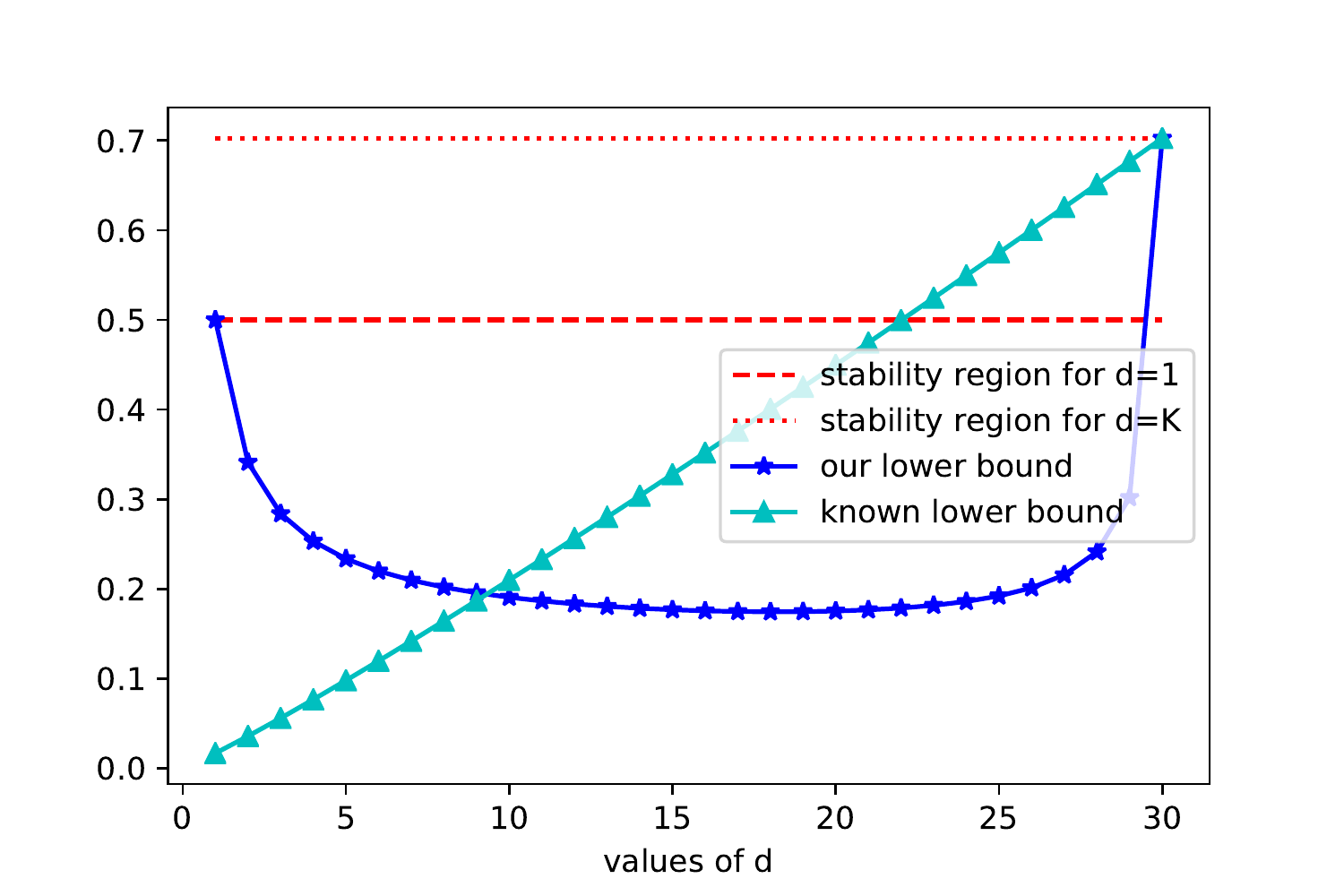}
\caption{\footnotesize Our lower bound and the known lower bound for $K=30$ and different values of $d$ for the case where $\mathbbm{E}[\wedge_{k=1}^dB_k]=2K/d^{1.1}$ . 
}
\label{fig:anchor_speed_random}
\end{subfigure}

%\caption{blah. }
\label{fig:anchor_speed}
\end{figure}
%%%%%%%%%%%%%%%%%%%%%%%%%%%%%%%%%%%%%
%%%%%%%%%%%%%%%%%%%%%%%%%%%%%%%%%%%%

\section*{Acknowledgment}
The author would like to thank Rami Atar, Isaac Keslassy and Shay Vargaftik for their useful feedback. This research was supported in part by the Hasso Plattner Institute.

\bibliographystyle{ieeetr}
\bibliography{bib}

\end{document}